\documentclass[letterpaper, 10 pt, conference]{ieeeconf}  
\usepackage{blindtext}
\usepackage{graphicx}

\IEEEoverridecommandlockouts                              

\usepackage{multirow}
\usepackage[latin1]{inputenc} 
\usepackage[cmex10]{amsmath}
\usepackage{amssymb}
\usepackage{amsmath}

\usepackage[cmex10]{amsmath}
\usepackage{amssymb}
\usepackage{multirow}
\usepackage[latin1]{inputenc} 
\usepackage{tikz}
\usetikzlibrary{shapes,arrows} 
\usepackage{amsmath}
\usepackage{enumerate}

\usepackage{proof}
\usepackage{algpseudocode}
\usepackage{upgreek}										
\newcommand{\numberset}{\mathbb}							

\newcommand{\R}{\numberset{R}}

\usepackage{mathtools}																	
\usepackage{booktabs}										
\usepackage{graphics} 										
\usepackage{epsfig} 										
\usepackage{contour}										
\usepackage{bm}
\usepackage{times} 		
\usepackage{rotating}										
\usepackage{longtable}					
\usepackage{pdflscape}	
\usepackage{cite} 									
\usepackage{xcolor} 
\usepackage{dsfont}		

\usepackage{centernot}
\usepackage{color}
\usepackage[acronym]{glossaries}
\usepackage[ruled]{algorithm2e}
\usepackage{graphicx,subcaption}
\usepackage{tabu}
\usepackage[official]{eurosym}
\usepackage{tikz}
\usetikzlibrary{shapes,arrows} 
\usepackage[american,cute inductors,smartlabels]{circuitikz}
\usetikzlibrary{arrows,automata}									
\usetikzlibrary{arrows,shapes,backgrounds,calc,positioning,patterns}
\usetikzlibrary{calc,arrows,shapes,backgrounds,calc,positioning,patterns,decorations.pathmorphing,decorations.markings,mindmap,trees}
\tikzstyle{block} = [draw, rectangle, minimum height=2em, minimum
width=4em] \tikzstyle{sum} = [draw, fill=blue!20, circle, node
distance=1cm] \tikzstyle{input} = [coordinate] \tikzstyle{output} =
[coordinate] \tikzstyle{pinstyle} = [pin edge={to-,thin,black}]
\usetikzlibrary{calc}
\ctikzset{bipoles/thickness=1} \ctikzset{bipoles/length=0.8cm}
\ctikzset{bipoles/diode/height=.375}
\ctikzset{bipoles/diode/width=.3}
\ctikzset{tripoles/thyristor/height=.8}
\ctikzset{tripoles/thyristor/width=1}
\ctikzset{bipoles/vsourceam/height/.initial=.7}
\ctikzset{bipoles/vsourceam/width/.initial=.7} \tikzstyle{every
	node}=[font=\small] \tikzstyle{every path}=[line width=0.8pt,line
cap=round,line join=round]

\newtheorem{theorem}{Theorem}
\newtheorem{lemma}{Lemma}
\newtheorem{proposition}{Proposition}

\newtheorem{definition}{Definition}
\newtheorem{assumption}{Assumption}
\newtheorem{remark}{Remark}
\DeclareMathOperator{\diag}{diag}
\DeclareMathOperator{\blockdiag}{blockdiag}

\usepackage{siunitx}										
\renewcommand{\vec}{\bm}

\DeclareMathAlphabet{\mathpzc}{OT1}{pzc}{m}{it}
\DeclareMathAlphabet{\mathcal}{OMS}{cmsy}{m}{n}

\DeclareMathOperator{\col}{col}
\newtheorem{objective}{Objective}

\usepackage{cases}
\usepackage{tikz}
\usetikzlibrary{shapes,arrows} 

\usepackage{array}

\usepackage{fixltx2e}

\usepackage{dblfloatfix}

\usepackage{url}

\usepackage{balance}

\usepackage{caption}
\usepackage{subcaption}
\usepackage{tikz,pgfplots}
\usepackage{epsfig}

\title{\LARGE \bf
Distributed Optimal Load Frequency Control with Stochastic  Wind Power Generation}

\author{Amirreza~Silani$^{1,2}$, Michele~Cucuzzella$^{1}$, Jacquelien~M.~A.~Scherpen$^{1}$ and Mohammad~Javad~Yazdanpanah$^{2}$ 
\thanks{*This work is supported by the EU Project MatchIT (project number: 82203).}
\thanks{$^{1}$A.~Silani, M.~Cucuzzella, and J.~M.~A.~Scherpen  are with the Jan~C.~Willems Center for Systems and Control, ENTEG, Faculty of
	Science and Engineering, University of Groningen, Nijenborgh 4, 9747
	AG Groningen, The Netherlands
	{\tt\small \{a.silani, m.cucuzzella, j.m.a.scherpen\}@rug.nl}.}%
\thanks{$^{2}$A.~Silani and M.~J.~Yazdanpanah are with the Control \& Intelligent Processing Center of Excellence, School of
	Electrical and Computer Engineering, University of Tehran, Tehran,
	Iran 
	{\tt\small \{a.silani, yazdan\}@ut.ac.ir}.}%
}

\begin{document}

\maketitle
\thispagestyle{empty}
\pagestyle{empty}

\begin{abstract}
Motivated by the inadequacy of conventional control methods for power networks with a large share of renewable generation, in this paper we study the (stochastic) passivity property of wind turbines based on the Doubly Fed Induction Generator (DFIG).
  Differently from the majority of the results in the literature, where renewable generation is ignored or assumed to be constant, we model wind power generation as a stochastic process, where
 wind speed is described by a class of stochastic differential equations. Then, we design a distributed control scheme that achieves load frequency control and economic dispatch, ensuring the stochastic stability of the controlled network.
\end{abstract}

\section{INTRODUCTION}

The supply-demand balance is an essential control objective in power networks. Indeed, the supply-demand mismatch  leads to frequency deviations from the nominal value, which eventually may result in stability disruptions \cite{ref1,ref2}. For this reason, the main control objective in power networks is the so-called Load Frequency Control (LFC).
 Additionally, another key objective is the minimization of the generation costs, also known as economic dispatch \cite{ref47}. The economic dispatch together with the LFC is called in the literature \emph{Optimal}
LFC (OLFC) (see for instance \cite{ref47,ref20,cai,ref49} and the references therein). However, due to the growing 
share of renewable generation sources in power networks, the conventional control schemes may be not adequate \cite{ref3}. 

Different control strategies achieving LFC and OLFC have been proposed for instance in \cite{ref39,ref40,ferrara} and \cite{ref47,ref49,ref26,ref46,ref48,ref50}, respectively (see also the references therein). However, in all these works, only conventional power generation is taken into account.

\subsection{Motivation and Contributions}
Nowadays, renewable generation sources are widespread in power networks, leading to an inevitable increase of uncertainties affecting the overall power system and its stability, resilience and reliability. For this reason, advanced control methods that guarantee the stability of the power system also in presence of time-varying renewable sources are necessary. Indeed, due to the random
and unpredictable nature of some primary energy sources such as wind, the dynamic behaviour of renewables can be usually described by stochastic processes (\emph{e.g.} Ito
calculus), as shown for instance in \cite{verdejo,muhando} for wind power generation. Also, \cite{ref51} proposes wind speed models based on Stochastic Differential Equations (SDEs), which can be useful in wind turbine models. Differently from \cite{ref39,ref40,ferrara,ref47,ref49,ref26,ref46,ref48,ref50} and other relevant works on the topic, in this paper we couple the wind speed model introduced in \cite{verdejo} with the model of wind turbines based on the Doubly Fed Induction Generator (DFIG). Then, we present a distributed passivity-based control scheme achieving OLFC and ensuring the stochastic stability of the power network.

The main contributions of this paper can be summarized as follows: (i)  the OLFC problem for nonlinear power networks including the turbine-governor model of conventional generators and the model of DFIG-based wind turbines is formulated, where the wind speed is modeled by an SDE; (ii) sufficient conditions for the stochastic passivity of the open-loop system are presented, facilitating the interconnection with passive control systems; (iii) a control scheme is proposed to obtain the passivity property of the DFIG-based wind turbine; (iv) the stochastic stability of the power network controlled by the distributed control scheme proposed in \cite{ref47} is proved and OLFC objective is achieved. 


\subsection{Notation}
The set of real numbers is denoted by $\R$. The set of positive (nonnegative) real numbers is denoted by $\R_{>0}$ ($\R_{\geq0}$).
Let $\boldsymbol{0}$ denote the vector of all zeros and the null matrix of suitable dimension(s), and $\boldsymbol{1}_n\in\R^n$ denote the vector containing all ones. The $n \times n$ identity matrix is denoted by $\mathds{I}_n$. Let $A \in \R^{n \times n}$ be a matrix. In case $A$ is a positive definite (positive semi-definite) matrix, we write $A > \vec{0}$  ($A \geq \vec{0}$). 
Let $\vert A\vert$ denote the matrix $A$ with all elements positive. 
The $i$-th element of vector $x$ is denoted by $x_i$. A steady-state solution to system $\dot x = f(x)$, is denoted by $\overline x$,  \textit{i.e.}, $\boldsymbol{0} = f(\overline x)$.
 Let $x\in\R^n,y\in\R^m$ be vectors,
  then we define $\col(x,y):=(x^\top~y^\top)^\top\in\R^{n+m}$.
Given a vector $x\in \R^n$, $[x]\in \R^{n\times n}$ indicates the diagonal matrix whose diagonal entries are the components of $x$ and $\sin(x):=\col\big(\sin(x_1),\dots,\sin(x_n)\big)$.  

\section{Problem Formulation}\label{sec:2}
In this section, we introduce the nonlinear power system model together with the turbine-governor and wind turbine models. Then, two control objectives are presented: load frequency control and   \emph{optimal} generation (economic dispatch). 
\subsection{Power Network Model}\label{power network}
In this subsection, we discuss the model of the considered power network (see Table~\ref{tab1} for the description of the symbols used throughout the paper). 
The  network topology is represented by an undirected and connected graph $\mathcal{G}=(\mathcal{V},\mathcal{E})$, where $\mathcal{V}=\mathcal{V}_{c}\cup \mathcal{V}_{w}=\{1, 2, ..., n\}$ is the set of the control areas and $\mathcal{E}=\{1, 2, ..., m\}$ is the set of the transmission lines. Specifically, the network comprises  $n_c$ conventional (synchronous) generators and $n_w$ wind turbine generators. Then, $\mathcal{V}_{c}=\{1, 2, ..., n_c\}$ is the set of the control areas including conventional (synchronous) generators and $\mathcal{V}_{w}=\{n_c+1, 2, ..., n\}$, with  $n=n_c+n_{w}$, is the set of the control areas including wind turbine generators.
Moreover, in analogy with \cite{ref16,ref17}, we assume that the power network is lossless and each node represents an aggregated area of generators and loads. Let $\mathcal{A} \in \R^{n \times m}$ denote the incidence matrix corresponding to the network topology. Then, the dynamics of the overall network (known as \emph{swing} dynamics) for all nodes (areas) $i \in \mathcal{V}$ are the following (see also \cite{ref16,ref17,ref47,ref49} for further details):
\begin{table}
	\centering
	\caption{Symbols}
	{\begin{tabular}{ll | ll}
			\toprule			
			
			$P_{c i}$		& Conventional power 	&$X_{m i}$      &  DFIG magnetizing\\
			~               &generation                  &~               & reactance\\
			$P_{w i}$		& Wind power generation &$X_{r i}$      &  DFIG rotor reactance\\  
			$P_{l i}$		& Unknown constant load &$X_{s i}$      &  DFIG Stator reactance\\
			$\varphi _i$	& Voltage angle &$X_{u i}$      &  Ratio between DFIG   \\
			$\omega_i$	    & Frequency deviation  &~ &magnetizing and \\
			$V_i$ 	    	& Voltage          &~  & stator self-inductance  \\
			$\iota_{{ds}_{i}}$ 	    	& $d$ component of DFIG  &$R_{r i}$      &  DFIG Rotor resistance \\
			~              &stator current &$R_{s i}$      &DFIG Stator resistance\\
			$\iota_{{qs}_{i}}$ 	    	& $q$ component of DFIG  &$\psi_i$        & Damping constant\\
			~                       &stator current  &$B$             & Susceptance\\
			$\iota_{{dr}_{i}}$ 	    	& $d$ component of DFIG  &$\bar{E}_{f i}$       & Exciter voltage\\
			~                       &rotor current  &$\tau_{c i}$    & Turbine time constant\\ 
			$\iota_{{qr}_{i}}$ 	    	& $q$ component of DFIG  &$H_{i}$         & Turbine inertia of \\
			~                       &rotor current &~  &wind turbine\\
			$V_{dr i}$ 	    	& $d$ component of DFIG &$T_{m i}$         & Mechanical torque of \\
			~                       & rotor voltage  &~  &wind turbine\\
			$V_{qr i}$ 	    	& $q$ component of  &$\lambda_i$         & Tip-speed ratio of \\			             
			~                       & DFIG rotor voltage   &~  &wind turbine\\
			$V_{t i}$    	    	& Terminal voltage &$r_i$          & Rotor radius of \\						
			~                       & of DFIG   &~   &wind turbine\\
			$f_{r i}$ 	    	& Rotor angular &$C_{Q i}(\lambda_i)$         & Power coefficient of   \\
			~                   & speed of DFIG       &~                  &wind turbine\\
			$f_{b i}$    	   	& Base speed   of DFIG                 &$\rho$         & Air density \\ 
			${v}_{i}$ 	    	& Predicted term of                     &$\xi_i$        &Speed regulation \\
			~                   &wind speed          &~          &coefficient\\
			$\tilde{v}_{i}$ 	   	& Stochastic term   &$\mathcal{N}_i$ & Neighboring areas \\
			~                       & of wind speed     &~                &of area $i$\\
			$\tau _{p i}$   & Moment of inertia           &$\mathcal{A}$        & Incidence matrix \\
			$\tau _{v i}$   & Direct axis transient               &~                     &of power network\\
			~               &open-circuit constant                &$L^{\mathrm{com}}$   & Laplacian matrix \\
			 $X_{d i}$       & Direct synchronous    &~                     &of communication \\
			 ~                & reactance           &$u_{ci}$	     	& Control input for \\
			$X'_{d i}$      & Direct synchronous    &~    &conventional generator\\
			~               &transient reactance  &$u_{wi}$	     	& Control input for  \\
			~                &~                   &~                &wind turbine\\

					
			\bottomrule
	\end{tabular}}
	\label{tab1}
	\vspace{-1em}
\end{table}

\begin{equation}\label{eq3}
	\begin{split}
		\dot{\theta}&=\mathcal{A}^\top \omega \\
		\tau _p \dot{\omega} &= -\psi \omega +
		P-P_l -\mathcal{A}\Upsilon(V) \sin(\theta)\\
		\tau _v \dot{V}&=-\chi _dE(\theta)V+\bar{E}_{f},\\
	\end{split}
\end{equation}
where $\omega, V : \R_{\geq 0}\rightarrow \R^{n}$, $P : \R_{\geq 0}\rightarrow \R^{n}$ is defined as $P:=\col(P_{c},P_{w})$, with $P_{c}: \R_{\geq 0}\rightarrow \R^{n_c}$, $P_{w}: \R_{\geq 0}\rightarrow \R^{n_w}$ denoting the vector of the power generated by conventional and wind turbine generators, respectively,  $\theta : \R_{\geq 0}\rightarrow \R^{m}$ denotes the vector of the voltage angles differences, $\chi_d\in \R^{n\times n}$ is a diagonal matrix whose diagonal elements are defined as $\chi _{d i}:=X _{d i}-X' _{d i}$, with $X _{d i}, X' _{d i}\in\R$,
$\tau _p, \tau _v, \psi, P_l \in \R^{n\times n}$, and $\bar{E}_f\in\R^n$. Moreover, $\Upsilon : \R^{n}\rightarrow \R^{m\times m}$ is defined as $\Upsilon(V):=\diag\{\Upsilon_1, \Upsilon_2, ..., \Upsilon_m\}$, with $\Upsilon_k:=V_iV_jB_{ij}$, where $k\sim\{i,j\}$ denotes the line connecting areas $i$ and $j$. Furthermore, for any $i,j\in \mathcal{V}$, the components of $E : \R^m\rightarrow \R^{n\times n}$ are defined as follows: 
\begin{align}
	\begin{split}
		E_{ii}(\theta)=&~\frac{1}{\chi_{d i}}-B_{ii}, \hspace{6.8em} i \in \mathcal{V}\\
		E_{ij}(\theta)=&-B_{ij}\cos(\theta_k)=E_{ji}(\theta), \hspace{1.2em} k \sim \{i,j\} \in \mathcal{E} \\
		E_{ij}(\theta)=&~0,  \hspace{10.35em}\text{otherwise}.
	\end{split}
\end{align}
\begin{remark}\textbf{(Susceptance and reactance).}
	According to \cite[Remark~1]{ref49}, we notice that the reactance $X_{di}$ of each generator $i\in\mathcal{V}$ is in practice generally larger than the corresponding transient reactance $X'_{di}$. Furthermore, the self-susceptance $B_{ii}$ is negative and satisfies $\vert B_{ii}\vert>\sum_{j\in \mathcal{N}_i} \vert B_{ij}\vert$. Therefore, $E(\theta)$ is a strictly diagonally dominant and
	symmetric matrix with positive elements on its diagonal, implying that $E(\theta)$
	is positive definite \cite{ref16}.
\end{remark}

\subsection{Turbine-Governor Model for Conventional (Synchronous) Generators}\label{turbine governor}
In this subsection, we introduce the dynamics of the turbine-governor typically coupled with conventional (synchronous) generators. 
Specifically, we express the power generated by the  (equivalent) synchronous generator $i\in \mathcal{V}_{c}$ as the output of a first-order dynamical system describing the behaviour of the turbine-governor, \textit{i.e.},
\begin{equation}\label{eq15}
	\begin{split}
		\tau_{c i}\dot{P}_{c i}=-P_{c i}-\xi_i^{-1}\omega_i +u_{ci}, 
	\end{split}
\end{equation}
where $u_{ci} : \R_{\geq 0}\rightarrow \R$ is the control input and $\tau_{c i},\xi_i\in\R_{>0}$. Now, we can write systems (\ref{eq15})  compactly for all nodes $i\in \mathcal{V}_{c}$ as 
\begin{equation}\label{governor}
	\tau_{c}\dot{P}_{c}=-P_{c}-\xi^{-1}\omega +u_{c},
\end{equation}
where  $u_{c} : \R_{\geq 0}\rightarrow \R^{n_c}$ and $\tau_{c},\xi\in \R^{n_c\times n_c}$. 

Now, as it is customary in the power systems literature (see for instance \cite{ref16,ref47,ref49}), we assign to the power generated by the synchronous generator $i\in \mathcal{V}_{c}$, the following strictly convex linear-quadratic cost function:  
\begin{equation}\label{eq11}
	\begin{split}
		J^c_{i}(P_{c i})=\dfrac{1}{2}q_i P_{c i}^2+z_iP_{c i}+c_i,
	\end{split}
\end{equation}
where $J^c_{i} : \R\rightarrow \R$, $q_i\in\mathbb{R}_{>0}$, $z_i\in\mathbb{R}$, and $c_i\in\mathbb{R}$ for all $i\in \mathcal{V}_{c}$.
\subsection{DFIG-Based Wind Turbine Generator Model}\label{wind turbine}
In this subsection, we introduce the Doubly Fed Induction Generator (DFIG) dynamics of a wind turbine generator. In the DFIG-based wind turbine generator, two back-to-back converters including a rotor side converter  and a grid side
converter are used. The rotor side converter controls the rotor current, while the grid side
converter controls the DC link voltage \cite{ref52,ref52.1}. Since wind speed affects the generated power of a wind turbine, it is then important to have a realistic model of the wind speed.  In our model, we consider that the wind speed at each node $i\in\mathcal{V}$ is given by the sum of a predicted constant component $v_i$ and a stochastic component $\tilde{v}_i$. For this reason, an appropriate mathematical framework such as the Ito calculus framework is adopted
to analyze the DFIG model with stochastic wind speed and to control the active power generated by the wind turbine. Before introducing the DFIG dynamics,  we recall for the readers' convenience the definition of stochastic differential equation through the Ito calculus framework \cite{ref25ch2,ref26.1ch2}.
\begin{definition}{\bf (Stochastic differential equation).} A stochastic differential equation (SDE) is defined as follows:
	\begin{equation}\label{eq17.9}
		dx(t)=f(x,u)dt+g(x)d\beta(t),
	\end{equation}
	where $f(x,u)\in\mathbb{R}^{N}$ and $g(x)\in\mathbb{R}^{N\times M}$ are locally Lipschitz, $x(t)\in\mathbb{R}^N$ is the state vector of the stochastic process, $u(t)\in\mathbb{R}^P$ is the input of the system and $\beta(t)\in \mathbb{R}^{M}$ is the standard Brownian motion vector. 
\end{definition}

Now, according to \cite{ref52,ref52.1}, the dynamics of the DFIG-based wind turbine generator $i\in\mathcal{V}_{w}$  are given by
\begin{equation}\label{eq19}
	\begin{split}
		\dot{\iota}_{ds i}=&~\dfrac{f_{b i}}{K_i}\Big(-R_{s i} X_{r i} \iota_{ds i}+(K_i +X_{m i}^2f_{r i})\iota_{qs i}+R_{r i}X_{m i}\\
		&~\iota_{dr i}+X_{m i}X_{r i}f_{r i}{ \iota_{qr i}}+X_{r i}V_{t i}-X_{m i}V_{dr i}\Big)\\
		:=&~h_{ds i}(x_i)+b_{{s}i}V_{dr i} \\
		\dot{\iota}_{{qs}{i}}=&~\dfrac{f_{b i}}{K_i}\Big(-(K_i +X_{m i}^2f_{r i})\iota_{ds i}-R_{s i}X_{r i} \iota_{qs i}\\
		&-X_{m i}X_{r i}f_{r i} \iota_{dr i}-X_{m i} V_{qr i}+R_{r i}X_{m i} \iota_{qr i}\Big)\\
		:=&~h_{qs{i}}(x_i)+{b_{{s}i}} V_{qr i} \\
		\dot{\iota}_{{dr}{i}}=&~\dfrac{f_{b i}}{K_i}\Big(R_{s i}X_{m i} \iota_{ds i}-X_{s i} X_{m i}f_{r i} \iota_{qs i}-R_{r i} X_{s i} \iota_{dr i}\\
		&+(K_i -X_{s i} X_{r i}f_{r i})\iota_{qr i}-X_{m i}V_{t i}+X_{s i}V_{dr i}\Big) \\ 
		:=&~h_{{dr}{i}}(x_i)+b_{r i} V_{dr i}\\ 
		\dot{\iota}_{qr i}=&~\dfrac{f_{b i}}{K_i}  \Big(X_{s i} X_{m i} f_{r i} \iota_{ds i}+R_{s i}X_{m i} \iota_{qs i}+(X_{s i}X_{r i}f_{r i} \\ 
		&- K_i)\iota_{dr i}-R_{r i}X_{s i} \iota_{qr i}+X_{s i} V_{qr i}\Big)\\ 
		:=&~h_{qr i}(x_i)+b_{r i} V_{qr i}\\ 
	\dot{f}_{r i}=&~\dfrac{1}{2 H_i}\Big(T_{m i}(\tilde{v}_i) - X_{m i} (\iota_{{ds}{i}} \iota_{{qr}{i}} - \iota_{{qs}{i}} \iota_{{dr}{i}})\Big)\\ 
	:=&~{h_{f r{i}}}(x_i) \\ 
	P_{w i}=&-X_{u i}\iota_{{qr}{i}}f_{r i}\\
	:=&~\zeta_i(x_i),
\end{split}
\end{equation}
where $\iota_{{ds}{i}},  \iota_{{qs}{i}}, \iota_{{dr}{i}}, \iota_{{qr}{i}}, V_{dr i}, V_{qr i}, f_{r i}, \tilde{v}_{i}, P_{w i} :\R_{\geq 0}\rightarrow\R$, $x_i : \R_{\geq 0}\rightarrow\R^6$ is the state vector of DFIG defined as $x_i:=\col(\iota_{{ds}{i}}, \iota_{{qs}{i}}, \iota_{{dr}{i}}, \iota_{{qr}{i}}, f_{r i}, \tilde{v}_{i})$,  and $b_{s i},b_{r i}\in\R$ are defined as $b_{s i}:=-\frac{f_{bi} }{K_i}X_{m i}$ and $b_{r i}:=\frac{f_{bi} }{K_i}X_{s i}$. Also, $h_{ds i},h_{qs{i}},h_{{dr}{i}},h_{qr i},h_{f r{i}},\zeta_i:\R^{6}\rightarrow\R$, $V_{t i}, f_{b i}, X_{m i}, X_{r i}, X_{s i}\in\R$, $R_{s i},R_{r i},H_{i}\in\R_{>0}$, and $K_i\in\R$ is defined as $K_i:=X_{s i}X_{r i}- X_{m i}^2$. Moreover, $T_{m i} : \R\rightarrow\R_{\geq 0}$ is defined as $T_{m i}(\tilde{v}_i):=\frac{1}{2}\rho\pi r_i^3C_{Q i}(\lambda_i)(v_{i}+\tilde{v}_{i})^2$ with $v_{i}\in\R$, $\lambda_i,\rho,r_i\in\R_{>0}$, $C_{Q i} : \R_{>0}\rightarrow\R_{>0}$. Now, let the stochastic term of wind speed $\tilde{v}_{i}$  be modeled by a SDE as in \cite{ref51}, \textit{i.e.},
\begin{equation}\label{eq20}
	d \tilde{v}_{i}=-\mu_{w i}\tilde{v}_{i}dt+\sigma_{w i}\tilde{v}_{i}d\beta,~\forall i\in\mathcal{V}_{w},
\end{equation}
where $\mu_{w i}$ and $\sigma_{w i}$ are positive constant parameters.
Then, we can rewrite (\ref{eq19}) and (\ref{eq20}) compactly for all nodes $i\in\mathcal{V}_{w}$ as
\begin{equation}\label{eq21}
	\begin{split}
		dx&=(H_{g}(x)+B_{u} u_{w})dt+G(x)d\beta(t) \\
		P_{w}&=\zeta(x),
	\end{split}
\end{equation}
where $x : \R_{\geq 0}\rightarrow\R^{6n_w}$ is defined as $x:=\col(x_1,\dots,x_{n_w})$, $u_{w} : \R_{\geq 0}\rightarrow\R^{2n_w}$ with $u_{wi} : \R_{\geq 0}\rightarrow\R^{2}$ defined as $u_{wi}:=\col(V_{dr i},V_{qr i})$, $\beta: \R_{\geq 0}\rightarrow\R^{6n_w}$ is the standard Brownian motion vector. Furthermore, $H_{g} : \R^{6n_w}\rightarrow\R^{6n_w}$ is defined as $H_g(x):=\col\big(H_{g 1},\dots, H_{g n_w}\big)$ with $H_{g i}(x_i):=\col\big(h_{ds i}(x_i),h_{qs i}(x_i), h_{dr i}(x_i), h_{{qr}{i}}(x_i),$ $h_{f r{i}}(x_i), -\mu_i\tilde{v}_{i}\big)$, $G : \R^{6n_w}\rightarrow\R^{6n_w\times 6n_w}$ is defined as $G(x):=\blockdiag\big(G_1,\dots,G_{n_w}\big)$ with $G_i(x):=\diag(0, 0, 0, 0, 0, \sigma_{wi}\tilde{v}_{i})$, $\zeta : \R^{6n_w}\rightarrow\R^{n_w}$ is defined as $\zeta(x):=\col(\zeta_1,\dots,\zeta_{n_w})$ 
and $B_u\in\R^{6n_w\times 2n_w}$ is defined as $B_u:=\blockdiag\big(B_{u 1},\dots,B_{u n_w}\big)$ with $B_{u i}:=\col\big(({b_{{s}i}}~~0), (0~~{b_{{s}i}}), ({b_{r i}}~~0), (0~~{b_{r i}}), \boldsymbol{0}_{2\times 2}\big)$.

Now, we assign to the power generated by the wind turbine $i\in\mathcal{V}_{w}$, the following strictly concave linear-quadratic utility function:  
\begin{equation}\label{eq12}
	J^{w}_i(P_{w i})=-\dfrac{1}{2}q_i P_{w i}^2+z_iP_{w i}+c_i,
\end{equation}
where $J^{w}_{i} : \R\rightarrow \R$, $q_i\in\mathbb{R}_{>0}$, $z_i\in\mathbb{R}$, and $c_i\in\mathbb{R}$ for all $i\in\mathcal{V}_{w}$. Note that $q_i$ and $z_i$ are selected in order to take into account the value of the maximum power that the wind turbine can generate given the predicted wind speed $v_i$.

\subsection{Control Objectives}

In this subsection, we introduce and discuss the main control objectives of this work.
The first objective concerns the asymptotic regulation of the  frequency deviation to zero, \textit{i.e.}, 
\begin{objective}\textbf{(Load Frequency Control).}\label{obj:1}
	\begin{equation}\label{obj1}
		\lim_{t\rightarrow \infty}\omega(t)=\boldsymbol{0}_n.
	\end{equation}
\end{objective}

Besides improving the stability of the power network by regulating the frequency deviation to zero, advanced control strategies additionally aim at reducing the costs associated with the power generated by the conventional synchronous generators and increasing the utilities associated with the power generated by the wind turbines.
Therefore, we introduce the following optimization problem: 
\begin{equation}\label{eq13}
	\begin{split}
		&\min_{P} ~J(P)\\ 
		&~\text{s.t.}~\sum_{i\in\mathcal{V}}{\bar{P}_{i}}-P_{l i}=0, 
	\end{split}
\end{equation}
where $J(P)=\sum_{i\in \mathcal{V}_{c}} J^c_{i}(P_{c i})-\sum_{i\in\mathcal{V}_{w}}J^{w}_i(P_{w i})=\frac{1}{2} P^\top QP+Z^\top P+\boldsymbol{1}_n^\top C$ with $J^c_{i}(P_{c i})$, $J^{w}_i(P_{w i})$ given by \eqref{eq11}, \eqref{eq12}, respectively, Also, $Q\in\R^{n\times n}$, $Z,C\in\R^n$  are defined as $Q:=\diag(q_1,\dots, q_{n_c},q_{n_c+1},\dots, q_{n_c+n_w})$, $Z:=\col(z_1,\dots, z_{n_c},-z_{n_c+1},\dots, -z_{n_c+n_w})$, $C:=\col(c_1,\dots,  c_{n_c}, -c_{n_c+1},\dots, -c_{n_c+n_w})$, respectively. 
In this regard, \cite[Lemma 2]{ref49}, \cite[Lemma~3]{ref16} show that it is possible to achieve zero steady-state frequency deviation and simultaneously minimize the objective function $J(P)$ in \eqref{eq13} when the load $P_l$ is constant. More precisely, when the load $P_l$ is constant, the \emph{optimal} value of $P$, which allows for zero steady-state frequency deviation and minimizes (at the steady-state) the objective function $J(P)$ in \eqref{eq13}, solving the optimization problem \eqref{eq13}, is given by:
\begin{equation}\label{optimal}
	P^\mathrm{opt}= Q^{-1}\Big(\frac{\boldsymbol{1}_n\boldsymbol{1}_n^\top (P_l+Q^{-1}Z)}{\boldsymbol{1}_n^\top Q^{-1}\boldsymbol{1}_n}-Z\Big),
\end{equation}  
where $P^\mathrm{opt}:=\col(P_c^\mathrm{opt},P_w^\mathrm{opt})$.
This leads to the second objective, \textit{i.e.}, minimization of the objective function $J(P)$ in \eqref{eq13}, which is also known in the literature as economic dispatch or \emph{optimal} generation \cite{ref16,ref49}. Then, the second goal concerning the economic dispatch or \emph{optimal} generation is defined as follows: 
\begin{objective}\textbf{(Economic dispatch).}\label{obj:2}
	\begin{equation}\label{obj2}
		\lim_{t\rightarrow \infty} P(t) = P^\mathrm{opt},
	\end{equation}	
\end{objective}
with $P^\mathrm{opt}$ given by \eqref{optimal}.


We assume now that there exists a (suitable) steady-state solution to the considered augmented power network model (\ref{eq3}), (\ref{governor}) and (\ref{eq21}).  
\begin{assumption}\textbf{(Steady-state solution).}\label{ass1}
	There exists a constant input $(\bar{u}_c,\bar{u}_w)$ and a steady-state solution $(\bar{\theta}, \bar{\omega}, \bar{V}, \bar{P}, \bar{x})$ to (\ref{eq3}), (\ref{governor}) and (\ref{eq21}) satisfying 	
	\begin{equation}\label{eq4}
		\begin{split}
			\boldsymbol{0}&=\mathcal{A}^\top\bar{\omega} \\
			\boldsymbol{0}&= -\psi\bar{\omega} +\bar{P}-P_l-\mathcal{A}\Upsilon(\bar{V}) \sin(\bar{\theta})\\
			\boldsymbol{0}&=-\chi_d E(\bar{\theta})\bar{V}+\bar{E}_{f}\\
			\boldsymbol{0}&=-\bar{P}_{c}-\xi^{-1}\bar{\omega}+\bar{u}_c\\
			\boldsymbol{0}&=\big(H_{g}(\bar{x})+B_{u} \bar{u}_w\big)dt+G(\bar{x})d\beta.
		\end{split}
	\end{equation}
	Additionally, \eqref{eq4} holds also when $\bar{\omega}=\boldsymbol{0}$ and $\bar{P} = P^\mathrm{opt}$, with $P^\mathrm{opt}$ given by \eqref{optimal}.
\end{assumption}

In the next section, we present the passivity properties for the power network, turbine-governor and wind turbine. Then, we design a control scheme for regulating the frequency in presence of \emph{stochastic} wind power generation. 
To this end, in analogy with \cite{ref49,ref16}, the following assumption is required:
\begin{assumption}\textbf{(Steady-state voltage angle and amplitude).}\label{ass2}
	The steady-state voltage $\bar{V}\in\R^n$ and angle difference $\bar{\theta}\in\R^m$ satisfy 
	\begin{equation}\label{eq5}
		\begin{split}
			&\bar{\theta} \in (-\frac{\pi}{2},\frac{\pi}{2})^m,\\
			\chi_{d}E(\bar{\theta})-\diag(\bar{V})^{-1}\vert \mathcal{A}&\vert\Upsilon(\bar{V})\diag(\sin(\bar{\theta})) \\
			\diag(\cos(\bar{\theta}))^{-1}&\diag(\sin(\bar{\theta}))\vert \mathcal{A}\vert^\top  \diag(\bar{V})^{-1}>0.
		\end{split}
	\end{equation}
\end{assumption}
Note that Assumption~\ref{ass2} is usually verified in practice, \textit{i.e.}, the differences in voltage
(angles) are small and the line reactances are greater than the
generator reactances \cite{ref49,ref16}.


\section{Optimal Load Frequency Control} \label{sec:3}
In this section, we present the passivity properties for the power network, turbine-governor and wind turbine. Then, we use such passivity properties for designing a controller achieving Objectives~\ref{obj:1} and \ref{obj:2}.

%
\subsection{Incremental Passivity of Power Network and Turbine-Governor}
In this subsection, we recall from the literature the incremental passivity of the power network model introduced in Subsection~\ref{power network} and the turbine-governor model introduced in Subsection~\ref{turbine governor}. 
In analogy with \cite[Lemma~2]{ref16}, \cite[Lemma~3]{ref47}, the incremental passivity of system  \eqref{eq3} is obtained via the following lemma. 
\begin{lemma}\textbf{(Incremental passivity of system  \eqref{eq3}).} \label{prop:1}
	Let Assumptions \ref{ass1}, \ref{ass2} hold. System (\ref{eq3}) is  incrementally passive  with respect to the storage
	function
	\begin{equation}\label{eq7}
		\begin{split}
			S_1 =&-\boldsymbol{1}_n^\top \Upsilon(V)\cos(\theta)+\boldsymbol{1}_n^\top \Upsilon(\bar{V})\cos(\bar{\theta}) + \dfrac{1}{2}V^\top  DV \\
			&-\big(\Upsilon(\bar{V})\sin(\bar{\theta})\big)^\top(\theta-\bar{\theta})-\bar{E}_{fd}(V-\bar{V}) \\
			&-\dfrac{1}{2}\bar{V}^\top  D\bar{V}+\dfrac{1}{2}(\omega-\bar{\omega})^\top  \tau_{p}(\omega -\bar{\omega}),
		\end{split}
	\end{equation}
	and supply rate $(\omega -\bar{\omega})^\top (P-\bar{P})$, where the steady-state solution $(\bar{\theta},\bar{V},\bar{\omega})$ satisfies  (\ref{eq4}) and $D$ is a diagonal matrix with $D_{ii}=\frac{1-B_{ii}(X_{di}-X'_{di})}{X_{di}-X'_{di}}$.
\end{lemma}
\begin{proof}
The proof follows from combining \cite[Lemma~2]{ref16} and \cite[Lemma~3]{ref47}. Specifically, under the Assumption~\ref{ass2}, the storage function \eqref{eq7} is a positive definite function and satisfies
		\begin{equation}\label{eq6}
		\begin{split}
			\dot{S}_1
			=&-(\chi_{d}E(\theta)V-\bar{E}_{fd})^\top \tau_v^{-1}(\chi_{d}E(\theta)V-\bar{E}_{fd})
			\\&+(\omega-\bar{\omega})^\top (P-\bar{P})-(\omega-\bar{\omega})^\top \psi(\omega-\bar{\omega}),
		\end{split}
	\end{equation}
	along the solutions to \eqref{eq3}. 
\end{proof}

Now, we consider the following controller proposed in  \cite{ref47,ref49} for the turbine-governor $i\in \mathcal{V}_{c}$  
\begin{equation}\label{eq16}
	\begin{split}
		\tau_{\delta i}\dot{\delta}_{ i} = &-\delta_{i}+P_{c i}, \\
		u_{ci} = &~\delta_i,
	\end{split}
\end{equation}
where $\delta_i : \R_{\geq 0}\rightarrow\R$ and $\tau_{\delta i}\in\R_{>0}$. Then, in analogy with \cite[Lemma~5]{ref47} the incremental passivity of system  \eqref{eq15} in closed-loop with \eqref{eq16} is obtained via the following lemma.

\begin{lemma}\textbf{(Incremental passivity of \eqref{eq15}, \eqref{eq16}).}\label{prop:2}
	Let Assumption~\ref{ass1} hold. System \eqref{eq15} with controller \eqref{eq16} is  incrementally passive  with respect to the storage
	function 
	\begin{equation}\label{17.5}
		\begin{split}
			S_{2 i}=&~\dfrac{\tau_{c i}\xi_i}{2}(P_{c i}-{P}^\mathrm{opt}_{c i})^2+\frac{\tau_{\delta i}\xi_i}{2}(\delta_i-\bar{\delta}_i)^2,
		\end{split}
	\end{equation}
	and supply rate $-(P_{c i}-{P}^\mathrm{opt}_{c i})\omega_{i}$, where  the steady-state solution $({P}^\mathrm{opt}_{c i},\bar{\delta}_i)$ satisfies \eqref{eq4} and
	\begin{equation}\label{eq17}
		\begin{split}
			0= &-\bar{\delta}_{i}+{P}_{c i}^\mathrm{opt},
		\end{split}
	\end{equation}
	with ${P}^\mathrm{opt}_{c i}$ given by \eqref{optimal}.
\end{lemma}
\begin{proof}
See \cite[Lemma~5]{ref47}.
\end{proof}

\subsection{Stochastic Passivity Property for DFIG-Based Wind Turbine}
In this subsection, we propose a new control scheme to control the active power generated by the DFIG-based wind turbine. Then, we show that the DFIG-based wind turbine \eqref{eq19}, \eqref{eq20} in closed-loop with the proposed controller is stochastically passive.
Before introducing the DFIG controller, we recall for the readers' convenience the definitions of Ito derivative and stochastic passivity through the Ito calculus framework \cite{ref25ch2,ref26.1ch2}.
\begin{definition}{\bf (Ito derivative).} Consider a storage function $S(x)$, which is twice continuously differentiable. Then, $\mathcal{L}S(x)$ denotes the Ito derivative of $S(x)$ along the SDE \eqref{eq17.9}, \textit{i.e.},
	\begin{equation}
		\label{eq8}
			\mathcal{L}S(x)=\dfrac{\partial S(x)}{\partial x}f(x,u)+\frac{1}{2}\text{tr}\{g^\top (x)\dfrac{\partial ^2 S(x)}{\partial x^\top \partial x}g(x)\}.
	\end{equation}
\end{definition}
\vspace{.1cm}
\begin{definition} {\bf (Stochastic passivity).}~Consider system \eqref{eq17.9} with output $y=\eta(x)$. Assume that the deterministic and stochastic terms of the SDE (\ref{eq17.9}) at the equilibrium point are identically zero, \textit{i.e.}, $f(\bar{x},\bar{u})=g(\bar{x})=\boldsymbol{0}$. Then, system (\ref{eq17.9}) is said to be stochastically passive with respect to the supply rate $u^\top y$ if there exists a twice continuously differentiable positive semi-definite storage function $S(x)$ satisfying
	\begin{equation}\label{eq10}
		\mathcal{L}S(x)\leq u^\top y,~ \forall (x,u) \in\mathbb{R}^N\times\mathbb{R}^P.
	\end{equation}
\end{definition}

Now, consider the following controller for the DFIG-based wind turbine generator $i\in\mathcal{V}_{w}$:  
\begin{subequations}\label{eq22}
	\begin{align}\nonumber
		V_{dr i}=&-L_i(x_i)\big(\bar{K}_{1 i}(x_i)+\bar{K}_{2 i}(x_i)+\bar{K}_{3 i}(x_i)+\bar{x}_i^\top \Pi_i\bar{x}_i\\ \label{eq22.1}
		&+\bar{x}_i^\top \Pi_i x_i+\bar{x}_i^\top \Psi_i H_{g i}(x_i)\big) \\  \label{eq22.2}
		V_{qr i}=&-L_i(x_i)\big(D_{1 i}(x_i)\omega_i + D_{2 i}(x_i)\delta_i + D_{3 i}(x_i)\big) \\  \label{eq22.3}
		\tau_{\delta i} \dot{\delta}_i=&-\delta_i+P_{wi},
	\end{align}
\end{subequations}
where 
\begin{equation}
	\begin{split} 
		L_i(x_i)=& \, \frac{X_{r i}}{X_{r i}(\iota_{dr i}-\bar{\iota}_{dr i})-X_{m i}(\iota_{ds i}-\bar{\iota}_{ds i})} \\ \nonumber
		\bar{K}_{1 i}(x_i)=&~\rho\pi r_i^2C_{Q i}((f_{r i}-\bar{f}_{r i})v_i^2+v_i(f_{r i} - \bar{f}_{r i})^2)\\ \nonumber
		&+(f_{r i}-\bar{f}_{r i})^2 \\ \nonumber
		\bar{K}_{2 i}(x_i)=&~\Big(\iota_{{ds} i}-\dfrac{X_{m i}}{X_{s i}}\iota_{dr i}\Big)V_{ti}+\omega_i(P_{w i}-{P}_{w i}^\mathrm{opt})\\ \nonumber
		\bar{K}_{3 i}(x_i)=&~2\bar{f}_{r i} X_{m i}(\iota_{ds i}\iota_{qr i}-\iota_{qs i}\iota_{dr i})\\ \nonumber 
		&+\Big(R_{r i}\frac{X_{m i}}{X_{r i}}+R_{s i}\frac{X_{m i}}{X_{s i}}\Big)\iota_{dr i}\iota_{ds i}\\
		&+\Big(R_{r i}\frac{X_{m i}}{X_{r i}}+R_{s i}\frac{X_{m i}}{X_{s i}}\Big)\iota_{qr i}\iota_{qs i}\\
		\Pi_i=& \, \diag\big(R_{s i}, R_{s i}, R_{r i}, R_{r i}, 0, 0\big)\\
		\Psi_i=& \, \diag\Big(\dfrac{K_i \bar{\iota}_{ds i}}{f_{b i}X_{r i}}, \dfrac{K_i \bar{\iota}_{qs i}}{f_{b i}X_{r i}}, \dfrac{K_i \bar{\iota}_{dr i}}{f_{b i}X_{s i}}, \dfrac{K_i \bar{\iota}_{qr i}}{f_{b i}X_{s i}}, 0, 0\Big)\\ \nonumber
		D_{1 i}(x)=&-X_{u i} \iota_{qr i}f_{r i} +X_{u i}\bar{\iota}_{{qr} i} \bar{f}_{r i},\\ \nonumber
		D_{2 i}(x)=&~(P_{w i}-{P}^\mathrm{opt}_{w i})\delta_i \\ \nonumber
		D_{3 i}(x)=&~(\delta_i - P_{w i})^2-(P_{w i}-{P}^\mathrm{opt}_{w i}){P}^\mathrm{opt}_{w i},
	\end{split}
\end{equation}

Note that the controller \eqref{eq22} requires the information of $\bar{x}_i$ and ${P}^\mathrm{opt}_{w i}$ which can be obtained by solving \eqref{eq4} and \eqref{optimal}, respectively.
In order to obtain the stochastic passivity of \eqref{eq19}, \eqref{eq20}, \eqref{eq22}, we need to consider the following assumptions on the wind turbine and speed.
\begin{assumption}\label{ass3} \textbf{(Condition on the rotational speed).}
	The rotational speed $f_{r i}$ of the wind tubine $i\in\mathcal{V}_{w}$ is bounded as $\vert f_{r i}\vert<\bar{\gamma}_i$, $\bar{\gamma}_i\in\R_{>0}$.
\end{assumption}

\begin{assumption}\label{ass4} \textbf{(Condition on the parameters of \eqref{eq20}).}
	The wind speed parameters in (\ref{eq20}) satisfies 
	\begin{equation}\label{eq24}
		\mu_{w i}+\bar{f}_{r i} > \frac{\sigma_{w_i}^2}{2}+v_{i}+\bar{\gamma}_i, \quad i\in\mathcal{V}_{w}.
	\end{equation}
\end{assumption}
\vspace{.1cm}

Note that Assumption~\ref{ass3} is true in practice, since the rotational speed of a wind turbine is limited by the mechanical characteristics of the turbine itself, which is indeed usually equipped with mechanical breaks that avoid high rotational speed.
Assumption~\ref{ass4} is instead a sufficient technical condition to establish the stochastic passivity of the wind turbine.

Now, the stochastic passivity of DFIG-based wind turbine dynamics \eqref{eq19}, with wind speed dynamics \eqref{eq20} and controller \eqref{eq22} is obtained via the following proposition.
\begin{proposition} \label{lemma 1}\textbf{(Stochastic passivity of  \eqref{eq19}, \eqref{eq20}, \eqref{eq22}).}
	Let Assumptions~\ref{ass3} and \ref{ass4} hold. System \eqref{eq19}, \eqref{eq20} in closed-loop with \eqref{eq22} is stochastically passive with respect to the storage function 
	\begin{equation}\label{eq25}
		\begin{split}
		S_{3i} =&~\frac{K_i}{2f_{b i} X_{r i}}\Big(({\iota_{ds i}}-{\bar{\iota}_{ds i}})^2 +
			({\iota_{qs i}}-{\bar{\iota}_{qs i}})^2
		\\&	+({\iota_{dr i}}-{\bar{\iota}_{dr i}})^2+
			({\iota_{qr i}}-{\bar{\iota}_{qr i}})^2\Big) +2H_i(f_{r i}-\bar{f}_{r i})^2\\
			&+\rho\pi r_i^3C_{Q i}\tilde{v}_{i}^2+\frac{\tau_{\delta i}}{2}(\delta_i-\bar{\delta}_i)^2,
		\end{split}
	\end{equation}
	and supply rate $-\omega_i(P_{w i}-{P}_{w i}^\mathrm{opt})$, where the steady-state solution $(\bar{x}_i,{P}_{w i}^\mathrm{opt},\bar{\delta}_i)$  satisfies 
	\eqref{eq4} and
	\begin{equation}\label{eq17.}
		\begin{split}
			0= &-\bar{\delta}_{i}+{P}_{w i}^\mathrm{opt},
		\end{split}
	\end{equation}
	with ${P}^\mathrm{opt}_{w i}$ given by \eqref{optimal}.
\end{proposition}
\begin{proof}
	The Ito derivative of the storage function \eqref{eq25} satisfies
	\begin{equation}
		\begin{split}
			\mathcal{L}S_{3i}
			=&-R_{s i}({\iota_{qs i}}-{\bar{\iota}_{qs i}})^2-R_{s i}({\iota_{ds i}}-{\bar{\iota}_{ds i}})^2\\ 
			&-R_{r i}({\iota_{qr i}}-{\bar{\iota}_{qr i}})^2-R_{r i}({\iota_{dr i}}-{\bar{\iota}_{dr i}})^2\\  
			&-(\delta_i - P_{w i})^2-\omega_i(P_{w i}-{P}_{w i}^\mathrm{opt})-(\delta_i -\bar{\delta}_i)^2 \\ 
			&-(f_{r i}-\bar{f}_{r i})^2-\rho\pi r_i^3C_{Q i}(\mu_{w_i} -\frac{\sigma_{w_i}^2}{2}-v_{i} \\
			&-(f_{r i}-\bar{f}_{r i}))\tilde{v}_{i}^2-\rho\pi r_i^3C_{Q i}v_{i}((f_{r i} -\bar{f}_{r i})+\tilde{v}_{i})^2,
		\end{split}
	\end{equation}
	along the solution to \eqref{eq19}, \eqref{eq20}, \eqref{eq22}. Then, we can conclude
	that $\mathcal{L}S_{3i}\leq -(P_{w i}-{P}^\mathrm{opt}_{w i})\omega_{i}$. 
\end{proof}

\subsection{Closed-loop analysis}
In this subsection, we show that the closed-loop system is stochastically stable, achieving Objectives~\ref{obj:1} and \ref{obj:2}. First, we recall the definition of (asymptotic) stochastic stability \cite{ref25ch2,ref26.1ch2}.
	\begin{definition} {\bf((Asymptotic) stochastic stability).}  System (\ref{eq17.9}) is (asymptotically) stochastically stable if a twice continuously differentiable positive definite Lyapunov function $S: \R^N\longrightarrow \R_{>0}$ exists such that $\mathcal{L}S$ is (negative definite) negative semi-definite.
\end{definition}
Now, in order to achieve Objective 2, we modify controllers \eqref{eq16} and \eqref{eq22.3} as follows (see \cite{ref47,ref49}):
\begin{subequations}\label{eq31}
	\begin{align} \nonumber
		\tau_{\delta_i} \dot{\delta}_i=&-\delta_i+P_{c i} \\ \label{eq31a}
		&-\xi_i^{-1}q_i\sum_{j\in\mathcal{N}^\mathrm{com}_i}\big(q_i\delta_i +z_i-(q_j\delta_j -z_j)\big),~\forall i\in\mathcal{V}_{c}\\ \nonumber
		\tau_{\delta_i} \dot{\delta}_i=&-\delta_i+P_{w i} \\ \label{eq31b}
		&-q_i\sum_{j\in\mathcal{N}^\mathrm{com}_i}\big(q_i\delta_i +z_i-(q_j\delta_j- z_j)\big),~\forall i\in\mathcal{V}_{w}
	\end{align}
\end{subequations}
where $\tau_{\delta i}$ is the design parameter and $\mathcal{N}^\mathrm{com}_i$ is the set of areas communicating with area $i$. 
The distributed controller \eqref{eq31} can be written compactly
for all $i\in\mathcal{V}$ as
\begin{equation}\label{eq31.}
	\tau_{\delta} \dot{\delta}=-\delta+P -\blockdiag(\xi^{-1},\mathds{I}_{n_w})Q L^\mathrm{com}(Q\delta + Z),
\end{equation}
where $\delta : \R_{\geq 0}\rightarrow \R^{n}$, $\tau_{\delta} \in\R^{n\times n}$ and  $L^{\mathrm{com}}\in\R^{n\times n}$ is the Laplacian matrix associated with a connected communication network. 
More precisely, the term $Q\delta + R$ in \eqref{eq31.} reflects the marginal cost associated with the objective function $J(P)$ in \eqref{eq13} and $L^{\mathrm{com}}(Q\delta + Z)$ represents the exchange of such information among the areas of the power network.
%
In the following theorem, we show that the closed-loop system \eqref{eq3}, \eqref{governor}, \eqref{eq21}, \eqref{eq22.1}, \eqref{eq22.2}, \eqref{eq31.} is stochastically stable and Objectives~\ref{obj:1} and \ref{obj:2} are
attained.
\begin{theorem}\textbf{(Closed-loop analysis).}
	Let Assumptions~\ref{ass1}--\ref{ass4} hold. Consider system \eqref{eq3}, \eqref{governor}, \eqref{eq21} with controller \eqref{eq22.1}, \eqref{eq22.2}, \eqref{eq31.}. Then, the solutions to the closed-loop system starting sufficiently close to $(\bar{\theta}, \bar{\omega}=\boldsymbol{0}, \bar{V}, {P}^\mathrm{opt}, \bar{x}, \bar{\delta})$ stochastically converge to the set where $\bar{\omega}=\boldsymbol{0}$ and $\bar{P}={P}^\mathrm{opt}$, with ${P}^\mathrm{opt}$ given by \eqref{optimal}, \textit{i.e.}, achieving Objectives~\ref{obj:1} and \ref{obj:2}. 
\end{theorem}
\begin{proof}
	Following Lemmas~\ref{prop:1}, \ref{prop:2} and Proposition~\ref{lemma 1}, we consider the storage function $S=S_1+S_2+S_3$,
	where $S_1$ is given in \eqref{eq7}, $S_2=\sum_{i\in\mathcal{V}_{c}} S_{2 i}$, with $S_{2i}$ given by \eqref{17.5}, and $S_3=\sum_{i\in\mathcal{V}_{w}} S_{3 i}$, with $S_{3 i}$ given by \eqref{eq25}. Now, the gradient  of $S$ is given by
	\begin{equation}\label{eq33}
	\begin{split}
		\nabla S=& 
		\col\Big(\Upsilon(V)\sin(\theta)-\Upsilon(\bar{V})\sin(\bar{\theta}),\chi_{d}E(\theta)V-\bar{E}_{fd}, \\
			&\tau_{p}(\omega-\bar{\omega}), 
            \tau_{\delta}\blockdiag(\xi, \mathds{I}_{n_w})(\delta -\bar{\delta}), \\
			&\blockdiag(\tau_{c} \xi,\boldsymbol{0}_{n_w\times n_w})(P-{P}^\mathrm{opt}), K [f_{b}]^{-1}X_r^{-1}
			 \\
			&(\iota_{ds}-\bar{\iota}_{ds}), K [f_{b}]^{-1}X_r^{-1}(\iota_{qs}-\bar{\iota}_{qs}), 
			K [f_{b}]^{-1}X_r^{-1} \\
			&(\iota_{dr}-\bar{\iota}_{dr}), K [f_{b}]^{-1}X_r^{-1}(\iota_{qr}-\bar{\iota}_{qr}),\\
			&4H(f_r -\bar{f}_r), 2\rho \pi  [r]^3[C_Q]\tilde{v}\Big). 
			\end{split}
	\end{equation}
	We can observe from \eqref{eq33} that $\nabla S$ evaluated at $(\bar{\theta}, \bar{\omega}=\boldsymbol{0}, \bar{V}, {P}^\mathrm{opt}, \bar{x}, \bar{\delta})$ is equal to zero. Then, the Hessian matrix of $S$ is given by
	\begin{equation}\label{eq34}
		\begin{split}
			\nabla^2 S=&\blockdiag\Big(\Lambda, \tau_{p}, \tau_{\delta}\blockdiag(\xi, \mathds{I}_{n_w}),  \\ &\blockdiag(\tau_{c} \xi, \boldsymbol{0}_{n_w\times n_w}), (T_m K)^\top,  \\ &(K[f_{b}]^{-1}X_r^{-1})^\top, (K[f_{b}]^{-1}X_r^{-1})^\top, \\ &(K[f_{b}]^{-1} X_s^{-1})^\top , (K[f_{b}]^{-1}X_s^{-1})^\top , 4H^\top, \\ &(2\rho \pi  [r]^3[C_Q])^\top \Big),
		\end{split}
	\end{equation}
	where 
	\begin{equation}\label{Lambda}
		\Lambda= \left(\begin{array}{cc}\Upsilon(V)\diag(\cos(\theta)) &\Omega^\top \\
			\Omega &\chi_{d}E(\theta)\end{array}\right),
	\end{equation}
	with $\Omega=(\diag(V))^{-1}\vert \mathcal{A}\vert\Upsilon(V)\diag(\sin(\theta))$. By virtue of Assumption~\ref{ass2},  we have $\Upsilon(V)\diag(\cos(\theta))>0$ for $\theta\in (-\frac{\pi}{2},\frac{\pi}{2})^m$. Then in analogy with \cite[Lemma~2]{ref16} and by using the Schure complement of \eqref{Lambda}, the matrix $\Lambda$ evaluated at $(\bar{\theta}, \bar{\omega}=\boldsymbol{0}, \bar{V}, {P}^\mathrm{opt}, \bar{x}, \bar{\delta})$ is positive definite if and only if
	\begin{equation}\label{Schur}
		\begin{split}
			&\chi_{d}E(\bar{\theta})-\diag(\bar{V})^{-1}\vert D\vert\Upsilon(\bar{V})\diag(\sin(\bar{\theta}))\\
			&\quad \quad \quad \; \; \diag(\cos(\bar{\theta}))^{-1}\diag(\sin(\bar{\theta}))\vert D\vert^\top \diag(\bar{V})^{-1}>0.
		\end{split}
	\end{equation}
	Thus, by virtue of Assumption~\ref{ass2}, it can be inferred from \eqref{eq34}--\eqref{Schur} that the Hessian matrix $\nabla^2 S$ evaluated at $(\bar{\theta}, \bar{\omega}=\boldsymbol{0}, \bar{V}, {P}^\mathrm{opt}, \bar{x}, \bar{\delta})$ is positive definite. Consequently, the storage function $S$ has a local minimum at $(\bar{\theta}, \bar{\omega}=\boldsymbol{0}, \bar{V}, {P}^\mathrm{opt}, \bar{x}, \bar{\delta})$. 
	
	Now, the Ito derivative of the storage function $S$ satisfies
	\begin{equation}\label{eq35}
		\begin{split}
			\mathcal{L}S=&-(\chi_{d}E(\theta)V-\bar{E}_{fd})^\top  \tau_v^{-1}(\chi_{d}E(\theta)V-\bar{E}_{fd})\\
			&-(\omega-\bar{\omega})^\top \psi(\omega-\bar{\omega})-(\delta-P)^\top  \vartheta(\delta-P)\\
			&-(\iota_{qs}-\bar{\iota}_{qs})^\top  R_s (\iota_{qs}-\bar{\iota}_{qs})-(\iota_{ds}-\bar{\iota}_{ds})^\top  R_s \\
			&(\iota_{ds}-\bar{\iota}_{ds})-(\iota_{qr}-\bar{\iota}_{qr})^\top  R_r (\iota_{qr}-\bar{\iota}_{qr}) \\ &-(\iota_{dr}-\bar{\iota}_{dr})^\top  R_r (\iota_{dr}-\bar{\iota}_{dr})-\tilde{v}^\top \rho\pi [r]^3[C_{Q}]\\&(\mu_w-\frac{1}{2}\sigma_w^\top \sigma_w-v-(f_r-\bar{f}_r))\tilde{v} -((f_r-\bar{f}_r)+\tilde{v})^\top\\& \rho\pi [r]^3[C_{Q}]v((f_r-\bar{f}_r)+\tilde{v})
			-\big(Q\delta +Z-(Q\bar{\delta}+\\&Z)\big)^\top L^\mathrm{com}\big(Q\delta +Z-(Q\bar{\delta}+Z)\big)
	-(\delta-\bar{\delta})^\top (\delta-\bar{\delta})
		\end{split}
	\end{equation}
	along the solution to \eqref{eq3}, \eqref{governor}, \eqref{eq21},  \eqref{eq22.1}, \eqref{eq22.2}, \eqref{eq31.}, where $\vartheta=\blockdiag\big(\xi, \mathds{I}_{n_w}\big)$. 
	Then, it follows that $\mathcal{L}S\leq 0$.  Thus, we can conclude that the solutions
	to the closed-loop system \eqref{eq3}, \eqref{governor}, \eqref{eq21},  \eqref{eq22.1}, \eqref{eq22.2}, \eqref{eq31.} are
	bounded.
	Moreover, according to LaSalle's invariance principle, these solutions stochastically converge to the largest invariant set contained in $\Lambda := \{{\theta}, {\omega}, {V}, {P}, {x}, {\delta}:  {\omega}=\boldsymbol{0}_n, \chi_{d}E(\theta)V=\bar{E}_{fd,}  P=\delta, x=\bar{x}, Q\delta+Z=Q\bar{\delta}+Z+d(t)\boldsymbol{1}_n\}$, where $d(t): \mathbb{R}_{\geq 0}\rightarrow\mathbb{R}$. 
	Then, we can obtain $\bar{\delta}+d(t)Q^{-1}\boldsymbol{1}_n={P}^\mathrm{opt}+d(t)Q^{-1}\boldsymbol{1}_n=\delta=P$. Hence, the behavior of the
	power network system \eqref{eq3} on the set $\Lambda$ can be described by
	\begin{equation}\label{eq37}
		\begin{split}
			\dot{\theta}=&~\boldsymbol{0}\\
			\boldsymbol{0}=&~{P}^\mathrm{opt}-\mathcal{A}\Upsilon(\overline{V}')\sin(\overline\theta')+d(t)Q^{-1}\boldsymbol{1}_n-P_l\\
			\boldsymbol{0}=&-\chi_{d}E(\overline\theta')\overline{V}'+\bar{E}_{fd},
		\end{split}
	\end{equation}
	where $\overline{V}'$ and $\overline\theta'$ are constants (possibly different from $\overline{V}$ and $\overline\theta$). 
	Moreover, since $\boldsymbol{1}_n^\top ({P}^\mathrm{opt}-P_l)=0$, $\boldsymbol{1}_n^\top \mathcal{A}=0$, and $Q^{-1}$ is a positive definite diagonal matrix, we can pre-multiply the second equation of \eqref{eq37} by $\boldsymbol{1}_n^\top$  and obtain $d(t)=0$. Thus, we have $\delta=\bar{\delta}$ and can then infer from \eqref{eq17}, \eqref{eq17.} that $\bar{P}={P}^\mathrm{opt}$. Therefore, the solutions to the closed-loop system \eqref{eq3}, \eqref{governor}, \eqref{eq21},  \eqref{eq22.1}, \eqref{eq22.2}, \eqref{eq31.}, starting sufficiently close to $(\bar{\theta}, \bar{\omega}=\boldsymbol{0}, \bar{V}, {P}^\mathrm{opt}, \bar{x}, \bar{\delta})$ stochastically converge to the set where $\bar{\omega}=\boldsymbol{0}$ and $\bar{P}={P}^\mathrm{opt}$ with ${P}^\mathrm{opt}$ given by \eqref{optimal}. 
	%
\end{proof}
\section{Simulation Results}\label{sec:4}
\begin{table}[b]
	\caption{Constant parameters of simulation} 
	\centering 
	\begin{tabular}{c p{1cm} p{1cm} p{1cm} p{1.2cm}} 
		\toprule
		Parameter  & Area 1 &  Area 2 &  Area 3  & Area 4 \\ [0.5ex]
		
		\midrule 
		$B_{ii}$ (p.u.)~         & -56.3       & -58.5      & -56.2      & -49.4 \\
		~$q_i~ (\frac{\$10^4}{h})$       & 5        & 4.5            & 5.5      & 1 \\
		$\tau_{v}$ (s)~ ~~        & 6.32        & 6.63     & 7.15      & 6.46 \\
		$X_{di}$ (p.u.)       & 1.76          & 1.81     & 1.87      & 1.91 \\
		$X'_{di}$ (p.u.)       & 0.27    & 0.17     & 0.23      & 0.35 \\
		$E_{fdi}$(p.u.)       & 3.85         & 4.43     & 3.96      & 3.88 \\ 
		$\tau_{p i}$ (p.u.)~     & 3.95         & 4.71     & 5.23      & 4.17 \\ 
		$\psi_i$ (p.u.)~~     & 1.82         & 1.61     & 1.33      & 1.55 \\ 
		$\tau_{c i}$ (s)~ ~~      & 7.2         & 6.8     & 8.9      & - \\ 
		$\tau_{\delta i}$ (s)~ ~~     & 0.2     & 0.2     & 0.2      & 0.2 \\
		~~~~~$\xi_i$ (Hz p.u.$^{-1}$)   & 0.73     & 0.73     & 0.73      & - \\ 
		$R_{s i}$ (p.u.)~     & -     & -     & -      & 0.031 \\ 
		$R_{r i}$ (p.u.)~     & -     & -     & -      & 0.025 \\
		$X_{s i}$ (p.u.)~     & -     & -     & -      & 3.62 \\
		$X_{r i}$ (p.u.)~     & -     & -     & -      & 3.61 \\
		$X_{m i}$ (p.u.)     & -     & -     & -      & 3.6 \\
		$H_i$ (p.u.)~~~     & -     & -     & -      & 3.2 \\
		$r_i$ (m)~~~~~     & -     & -     & -      & 42 \\
		$\mu_{wi}$ (p.u.)~     & -     & -     & -      & 17.15 \\
		$\sigma_{wi}$ (p.u.)~     & -     & -     & -      & 2.65\\ [1ex] 
		\bottomrule 
	\end{tabular}
	\label{tab2}
\end{table} 
\begin{figure}[t]
	\centering
	\includegraphics[width=\columnwidth]{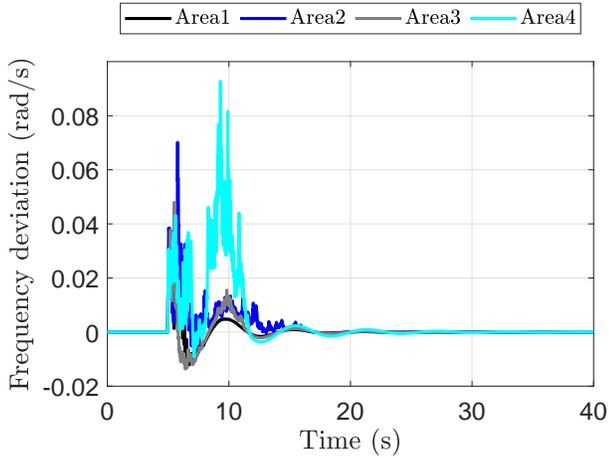}
	\caption{Frequency deviation in each area.}
	\label{f3}
\end{figure}
\begin{figure}[t]
	\centering
	\includegraphics[width=\columnwidth]{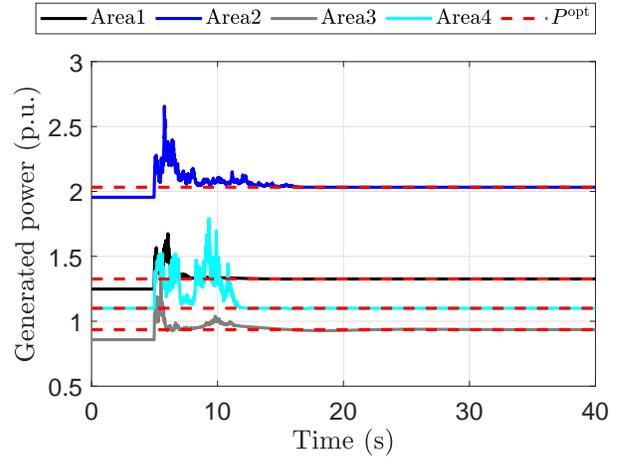}
	\caption{Generated power in each area.}
	\label{f2}
\end{figure}

\begin{figure}[t]
	\centering
	\includegraphics[width=\columnwidth]{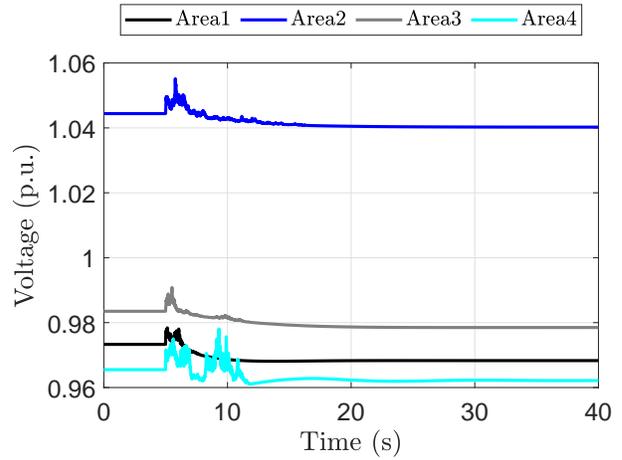}
	\caption{Voltage in each area.}
	\label{f4}
\end{figure}
In this section, the simulation results show excellent performance of the proposed distributed control scheme. We consider a power network partitioned into four control areas that are interconnected as represented in \cite[Fig.~1]{tcst2021}, where areas 1, 2 and 3 include conventional generation, while area 4 includes wind generation. We provide the system parameters
in Table~\ref{tab2}, where the parameters are equal to \cite[Table~I]{ref52} and \cite[Table~II]{tcst2021}, the nominal frequency and power base are
chosen equal to $120\pi$ rad$/$s and $1000$ MVA, respectively.


The system is initially at the steady-state with constant load $P_l=\col(1.3, 2, 1.3, 0.5)$. Then,  at the time instant $t=5$~s the load increases to $P_l=\col(1.4, 2.1, 1.4, 0.55)$ and the wind speed varies according the stochastic differential equation \eqref{eq20}. Fig.~\ref{f3} shows that the frequency deviation in each area converges to zero after a transient time. 
Also, we notice from Fig.~\ref{f2} that after $t=5$~s the generated power in each area converges to the corresponding optimal value (dashed line), which has been computed according to \eqref{optimal} with $P_l=\col(1.4, 2.1, 1.4, 0.55)$.
Specifically, we observe that the additional power demand is supplied by the conventional generators while the wind turbine (Area 4) generates the maximum possible power given a certain wind speed.
Moreover,  we can notice from Fig.~\ref{f4} that the voltages are stable.
\section{Conclusion}\label{sec:5}
In this paper, we have considered a power network including conventional synchronous generators with turbine-governor and wind turbines based on the doubly fed induction generator, where the wind speed is described by a stochastic differential equation. Then, we have verified the (stochastic) passivity of the considered system and present a distributed control scheme that guarantees the stochastic stability of the overall system, achieving optimal load frequency control.




%



\balance

\end{document}